\documentclass[11pt]{article}

\usepackage[margin=1in]{geometry}
\usepackage[T1]{fontenc}
\usepackage[utf8]{inputenc}

\usepackage{amsmath,amssymb,amsthm,mathtools}

\usepackage[numbers]{natbib}
\usepackage[hidelinks]{hyperref}

\newtheorem{theorem}{Theorem}
\newtheorem{corollary}{Corollary}
\newtheorem{definition}{Definition}
\newtheorem{remark}{Remark}

\DeclareMathOperator{\TV}{TV}

\DeclareMathOperator{\Prb}{Pr}
\DeclareMathOperator{\E}{\mathbb{E}}

\newcommand{\SOTM}{\mathcal{M}}
\newcommand{\Oracle}{\mathcal{O}}
\newcommand{\Dist}{\mathcal{D}}

\newcommand{\TOK}{\mathrm{TOK}}
\newcommand{\tok}{\mathrm{tok}}
\newcommand{\kappaT}{\kappa_T}

\title{\textbf{Computing with Stochastic Oracles in AI-Augmented Computation}%
}
\author{Jie Wang\,\thanks{Richard A. Miner School of Computing and Information Sciences,
University of Massachusetts, Lowell, MA 01854, USA.}\\[0.5ex]
\small\texttt{Jie\_Wang@uml.edu}}
\date{}

\begin{document}
\maketitle

\begin{abstract}
The Stochastic-Oracle Turing Machine (SOTM) framework models AI-augmented
computation as the interaction of a probabilistic Turing machine with an oracle
whose responses are drawn from context-dependent distributions. This paper
studies what an SOTM can achieve under two oracle-response schemes: in a
cached-response oracle, each distinct query receives one response that is reused
on later calls to the same query, while in a fresh-response oracle, each call
returns an independent response. In both schemes, the SOTM first computes from
its input and internal random source to generate its first query, then proceeds
adaptively, computing from its query-response transcript (the record of queries
issued and responses received) to generate each subsequent query or produce a
final output. Cached responses impose two transcript-based ceilings on achievable
performance: a correct-identification ceiling governed by the total variation
distance between the transcript distributions induced by the hidden states of the
oracle, and an output quality ceiling equal to the expected score of the best
output the SOTM can compute from the transcript. Fresh responses can raise
these ceilings by allowing repeated calls to accumulate independent evidence
toward correct or high-quality outputs. In the binary single-informative-query
case, the error probability decreases exponentially in the number of calls to the
same query at the Chernoff rate. For output quality, query-count bounds characterize threshold stopping when the
score function is incorporated as part of the SOTM, and majority-based
amplification bounds characterize the binary candidate-output model when it is
not. Together, the results identify how
response reuse, transcript information, and access to the score function
determine what an SOTM can compute and at what token cost.
\end{abstract}

\section{Introduction}

AI-augmented computing delegates knowledge-intensive or skill-intensive subtasks
to stochastic systems. A program may ask a general-purpose large language model
(LLM), a domain-specific fine-tuned model, or more generally a family of AI
systems with a built-in control mechanism, to solve a problem instance, produce
code, or judge a candidate solution, among other tasks.
Wang~\cite{Wang2026} introduced the Stochastic-Oracle Turing Machine (SOTM)
framework\footnote{Wang originally called the model the AI-Oracle Turing Machine
(AOTM)~\cite{Wang2026}, and later renamed it the Stochastic-Oracle Turing
Machine~\cite{Wang2026cert}. The term ``AI'' is context-dependent, while
``stochastic'' names the mathematically defined feature of the oracle.} to
formalize this paradigm, in which a probabilistic Turing machine interacts with
a stochastic oracle through a query-response interface, with responses drawn from
query-dependent distributions. Wang also introduced token complexity as an exact measure of the token cost for
solving a given task at a given quality level: the minimum expected weighted sum
of query and response tokens, with query-token weight $\alpha>0$ and
response-token weight $\beta>0$. In what follows, ``stochastic oracle'' is abbreviated as
``oracle.''

In practice, a system engineer may choose an LLM with strong benchmark
performance as an oracle and build a system to interact with it for completing a
given task. Benchmark performance alone, however, may not reveal the response
distributions that matter on the task instances at hand. Even if those
distributions are known or approximately modeled, one may still want to know what
an adaptive SOTM can achieve from interacting with the chosen oracle to inform
the design of the system. This paper studies that question under two
oracle-response schemes. Under cached responses, each distinct query receives one
response that is reused on later calls to the same query. Under fresh responses,
repeated calls to the same query return independent responses drawn from the
oracle.

The distinction matters because LLM responses may be wrong, misleading, or
ambiguous at times---the hallucination problem inherent in stochastic systems.
The SOTM computes from the responses it receives, but it may not have access to
the score function to evaluate whether its final output is correct or high
quality. Two issues arise. The first is identification accuracy: the
transcript of these query-response interactions (the record of queries issued
and responses received) may not let the SOTM distinguish which hidden state of
the oracle generated the responses. The second is output quality: even when no
hidden state needs to be identified, the transcript may not contain enough useful
information to compute a high-scoring output. Cached responses make both
issues more acute, since repeating a query returns the same response and adds
no new
information.

The cached-response results express both issues as transcript-based ceilings.
For binary identification, each hidden state induces a distribution on the finite
transcripts produced by the SOTM. The optimal success probability is governed by
the total variation distance between these transcript distributions; their
overlap is exactly the irreducible ambiguity that no adaptive strategy can
remove. This result extends to finite hidden-state families, and the
single-informative-query case gives a closed-form expression in terms of the
total variation distance between the two state-conditional response distributions
of that query.

For output quality, the question is not whether the SOTM identifies a hidden
state, but what quality it can compute from the transcript. The output
quality ceiling is the expected score of the best output computable from the
transcript. Unlike the identification issue, which concerns ambiguity about
the hidden state, the output quality issue can arise even when the hidden
state is fully determined: an SOTM may exhaust all information in its transcript
and still fail to produce a high-scoring output.

Fresh responses raise both ceilings by allowing repeated oracle calls to add
independent evidence about the hidden state. For identification, the
identification error in the binary single-informative-query setting decreases
exponentially in the number of queries, at a rate given by the Chernoff
information between the two response distributions. For output quality, fresh
responses convert the ceiling into a query-count problem. When the score function
is incorporated into the SOTM, the SOTM can keep querying until it finds an
output whose score reaches the target level. When the score function is not
available during computation, fresh responses can still help in a binary
candidate-output model through majority voting over repeated responses.

The results are distributional: they characterize what is achievable given the
response distributions of the oracle, not just its average behavior. Total
variation distance measures transcript distinguishability under cached responses;
expected score measures the quality that the transcript can support;
Chernoff information measures the asymptotic rate at which independent fresh
responses reduce binary identification error; and Bernoulli KL divergence
appears in the majority-based query-count bound in the binary candidate-output
model.

The paper is organized as follows. Section~\ref{sec:preliminaries} recalls the
SOTM model and the information-theoretic quantities used in the paper.
Section~\ref{sec:hidden} introduces the hidden-state abstraction and the
transcript framework. Sections~\ref{sec:identification} and~\ref{sec:coverage}
establish the cached-response identification and output quality ceilings.
Section~\ref{sec:fresh} studies fresh responses, including identification via
evidence accumulation and query-count bounds for output quality.
Section~\ref{sec:conclusion} concludes and lists extensions.

\section{Preliminaries}
\label{sec:preliminaries}

We build on the SOTM model and token-complexity terminology of
Wang~\cite{Wang2026}, generalizing the oracle's response distributions from
query-dependent to context-dependent as described below. A fixed finite
alphabet $\Sigma$ is understood, and all queries and responses are strings in
$\Sigma^\ast$.

\begin{definition}[Stochastic-Oracle Turing Machine (SOTM)]
\label{def:sotm}
A \textbf{Stochastic-Oracle Turing Machine} is a pair $\SOTM=(M,\Oracle)$,
where $M$ is a probabilistic Turing machine and $\Oracle$ is a stochastic
oracle. A computation of $\SOTM$ on an input $x\in\Sigma^\ast$ proceeds in
turns. Each turn $i$ consists of two phases: first, the \textbf{inter-turn
computation} phase, in which $M$ reads $x$ and the prior transcript $t_{i-1}$,
uses its internal state and random source, and either generates the next query
$q_i\in\Sigma^\ast$ or halts and writes an output $y\in Y$ to an output tape;
and second, the oracle exchange phase, in which $\Oracle$ draws a response $r_i$
from $\Dist_{q_i}^{x,t_{i-1}}$, the oracle's response distribution for query
$q_i$ in context $(x,t_{i-1})$, and returns it to $M$ (this phase is skipped if
$M$ halted in phase 1). Here $t_{i-1}=(q_1,r_1,\ldots,q_{i-1},r_{i-1})$ is the
realized prior transcript (with $t_0$ the empty transcript), where each pair
$(q_j,r_j)$ is the query and response from turn $j$. To distinguish random
variables from their realizations, we write $Q_j$ and $R_j$ for the random
query and response on turn $j$, and
$\mathcal{T}_{i-1}=(Q_1,R_1,\ldots,Q_{i-1},R_{i-1})$ for the random transcript
before turn $i$. The oracle is the family of context-dependent response
distributions
\[
  \Oracle = \bigl\{\Dist_q^{x,t} : q\in\Sigma^\ast,\, x\in\Sigma^\ast,\,
  t=(q_1,r_1,\ldots,q_{k},r_{k})\text{ a finite realized transcript}\bigr\}.
\]
The oracle's response distributions $\Dist_q^{x,t}$ encode knowledge or
capabilities not available to $M$ as a stand-alone probabilistic Turing machine:
$M$ cannot produce responses with the same distribution from $x$ and $t$ alone.
\end{definition}

\begin{remark}[Extension from query-dependent to context-dependent distributions]
Wang~\cite{Wang2026} originally defined the oracle's response distribution as
depending only on the query $q$, giving a family $\{\Dist_q:q\in\Sigma^\ast\}$.
This paper generalizes to context-dependent distributions
$\Dist_q^{x,t_{i-1}}$, where the distribution also depends on the input $x$ and
the realized prior transcript $t_{i-1}$. This extension is motivated by modern
LLM deployments, in which the same query posed after different prior exchanges,
or about different inputs, may activate different internal mechanisms---such as
different expert combinations in a mixture-of-experts
architecture~\cite{Shazeer2017,Fedus2022}---and
therefore produce responses from different distributions. The query-dependent
model of Wang~\cite{Wang2026} is the special case in which
$\Dist_q^{x,t_{i-1}}=\Dist_q$ for all $x$ and $t_{i-1}$.
\end{remark}

Throughout this paper we consider only SOTMs $\mathcal{M}=(M,\mathcal{O})$ such
that for every input $x\in\Sigma^\ast$, the computation halts with probability
one and produces a finite query-response transcript.

\begin{definition}[Task and Quality]
A task is $T=(X,Y,S,\Dist_X)$, where $X\subseteq\Sigma^\ast$ is an input space,
$Y$ is an output space, $S:X\times Y\to[0,1]$ is a score function, and $\Dist_X$
is a distribution on $X$. An SOTM $\SOTM=(M,\Oracle)$ \textbf{achieves
quality~$\theta$} on task $T$ if
\[
  \E_{x\sim\Dist_X}\bigl[S(x,\SOTM(x))\bigr] \;\geq\; \theta,
\]
where the expectation is over $x\sim\Dist_X$, the internal randomness of $M$,
and the response distributions of $\Oracle$, and $\SOTM(x)\in Y$ is the output
$M$ produces on input $x$.
\end{definition}

\begin{remark}[Availability of the score function]
Whether the score function $S$ is available to the SOTM during its computation
depends on the setting. In everyday LLM use---such as chatbots, search, or
summarization---there is typically no formal score function available during
computation; quality is judged informally by the user after the fact. In
specific vertical applications---such as code generation evaluated by test
suites, medical diagnosis with ground-truth labels, or legal document review
with compliance criteria---an explicit score function can be incorporated into
the SOTM, allowing it to evaluate and act on candidate outputs during its
computation. Both cases are studied in this paper; each result states explicitly
which applies.
\end{remark}

\begin{definition}[Token Counts, Token Cost, and Token Complexity]
For a computation of $\mathcal{M}$ on input $x$, let $N(x)$ denote the number of
oracle turns, and let $\tok_{\mathcal{M},\mathrm{in},i}(x)$ and
$\tok_{\mathcal{M},\mathrm{out},i}(x)$ denote the query-token and response-token
counts on turn $i$. The \textbf{token count} on turn $i$ is
\[
  \tok_{\mathcal{M},i}(x)
  =
  \tok_{\mathcal{M},\mathrm{in},i}(x)
  +
  \tok_{\mathcal{M},\mathrm{out},i}(x).
\]
Given a query-token weight $\alpha>0$ and a response-token weight $\beta>0$,
the \textbf{token cost} on turn $i$ is
\[
  \tok_{\mathcal{M},i}^{(\alpha,\beta)}(x)
  =
  \alpha\,\tok_{\mathcal{M},\mathrm{in},i}(x)
  +
  \beta\,\tok_{\mathcal{M},\mathrm{out},i}(x).
\]
The total token cost on input $x$ is
\[
  \TOK_{\mathcal{M}}^{(\alpha,\beta)}(x)
  =
  \sum_{i=1}^{N(x)} \tok_{\mathcal{M},i}^{(\alpha,\beta)}(x),
\]
and the expected token cost is
\[
  \TOK_{\mathcal{M}}^{(\alpha,\beta)}
  =
  \E_{x\sim\Dist_X}\!\left[\TOK_{\mathcal{M}}^{(\alpha,\beta)}(x)\right],
\]
where the expectation is over the input distribution, the internal randomness of
$M$, and the response distributions of $\Oracle$. The \textbf{token complexity}
of task $T$, oracle $\Oracle$, quality target $\theta$, and weights
$\alpha,\beta>0$ is
\[
  \kappaT(\theta;\Oracle,\alpha,\beta)
  =
  \inf_{\mathcal{M}}\;\TOK_{\mathcal{M}}^{(\alpha,\beta)},
\]
where the infimum is over all SOTMs $\mathcal{M}=(M,\Oracle)$ that achieve
quality~$\theta$ on $T$.
\end{definition}

\begin{remark}[Notational change from Wang~\protect\cite{Wang2026}]
Wang~\cite{Wang2026} used $\tok_{\mathcal{M},Q,i}$ and $\tok_{\mathcal{M},R,i}$
to denote query-token and response-token counts, with $Q$ and $R$ as subscripts.
This paper introduces random variables $Q_i$ and $R_i$ for the query and
response on turn $i$, creating a notational conflict with those subscripts. We
therefore separate token counts $\tok_{\mathcal{M},i}$ from token costs
$\tok_{\mathcal{M},i}^{(\alpha,\beta)}$, using $\mathrm{in}$ and $\mathrm{out}$
subscripts to distinguish query and response token counts. The token complexity
$\kappaT(\theta;\Oracle,\alpha,\beta)$ is otherwise defined identically to
Wang~\cite{Wang2026}.
\end{remark}

Token cost and output quality are competing quantities: more oracle turns consume
more tokens but may yield higher-quality outputs, while fewer turns reduce token
cost but may constrain achievable quality. The \textbf{complexity frontier} for
a task $T$ and oracle $\Oracle$ is the set of achievable (token cost, quality)
pairs
\[
  \mathcal{F}(T,\Oracle,\alpha,\beta)
  =
  \bigl\{
    \bigl(\TOK_{\mathcal{M}}^{(\alpha,\beta)},\,\theta\bigr)
    :
    \mathcal{M}\text{ achieves quality~}\theta\text{ on }T
  \bigr\},
\]
where the pair traces out a curve as $\theta$ varies. The token complexity
$\kappaT(\theta;\Oracle,\alpha,\beta)$ is the minimum token cost on this
frontier at quality level $\theta$, as introduced in Wang~\cite{Wang2026}. The
results of this paper characterize parts of this frontier: the cached-response
ceilings (Sections~\ref{sec:identification} and~\ref{sec:coverage}) give the
maximum quality achievable at the lowest token cost under cached responses, while
the fresh-response query-count bounds (Section~\ref{sec:fresh}) characterize how
token cost grows as the quality target $\theta$ increases.

\begin{definition}[Cached-Response and Fresh-Response Oracles]
A \emph{cached-response} oracle, on the first call to query $q$ at turn $i$,
draws a response from $\Dist_q^{x,t_{i-1}}$ and returns that same response on
every later call to $q$ at any turn $j>i$. A \emph{fresh-response} oracle draws
an independent response from $\Dist_q^{x,t_{i-1}}$ on every call to $q$ at
turn $i$; since later calls to the same query occur at different turns with
different prior transcripts, they draw from potentially different distributions.
\end{definition}

\begin{remark}[Query counts under the two schemes]
Query counts matter under both schemes, since they determine token costs
directly. Under cached responses, the query count is determined by the algorithm
for solving the task: each distinct query is issued at most once, since repeating
a query returns the same cached response and adds no new information. The token
cost analysis for cached responses therefore focuses on what quality the
transcript can support, rather than on the number of queries. Under fresh
responses, repeating the same query does add information, since each call returns
an independent response. This makes query count a meaningful quantity:
Section~\ref{sec:fresh} characterizes how many fresh queries are needed to reach
a target quality level, showing that repeated calls allow the SOTM to accumulate
evidence and produce better outputs than a single cached response permits.
\end{remark}

\begin{definition}[Output Functions, Per-Input Quality, and Competence]
\label{def:competence}
Given an input $x\in X$ and a realized transcript
$t=(q_1,r_1,\ldots,q_k,r_k)$ as defined in Definition~\ref{def:sotm}, the
remaining inter-turn computation of $\SOTM$ can be represented by an
\textbf{output function} $\phi:X\times(\Sigma^\ast)^\ast\to Y$ that produces
output $y=\phi(x,t)\in Y$. The \textbf{per-input expected quality} of $\phi$
under oracle $\Oracle$ on input $x$ is $\E[S(x,\phi(x,\mathcal{T}))]$, the
expectation over the random transcript $\mathcal{T}$ produced by the
computation. For a $0/1$ score function and a single-turn computation, the
\textbf{per-input competence} is
\[
  q(x)=\Prb[S(x,R_1)=1],
\]
where $R_1$ is one fresh oracle response.
\end{definition}

For distributions $\mu$ and $\nu$ on a countable set, their total variation
distance is
\[
  \TV(\mu,\nu)=\frac12\sum_r |\mu(r)-\nu(r)|.
\]
For Bernoulli parameters $a,b\in(0,1)$,
\[
  D(a\|b)
  =
  a\ln\frac{a}{b}
  +(1-a)\ln\frac{1-a}{1-b}
\]
is the Bernoulli KL divergence. Total variation governs the optimal success probability in binary identification
from a single observation; KL divergence governs the
exponential error rate in repeated independent testing~\cite{CoverThomas2006}.

\section{Hidden States and Computational Bounds}
\label{sec:hidden}

Because the oracle's response to a query is drawn from a response distribution,
the response need not be correct, so an SOTM cannot simply trust any single
response it receives. Nor can the SOTM check its output against the score
function, since the score function is unavailable during the computation unless
incorporated into the SOTM. It must instead use the pattern of oracle responses
to produce an output that scores well under the task's score function.

Stochastic oracle responses motivate the hidden-state abstraction. A hidden
state is an unobserved condition, determined by the oracle's internal mechanism
and the current context $(x,t_{i-1})$, that determines the active response
distribution for a query. The SOTM does not observe this condition directly; it
only observes the response drawn from the active distribution. Hidden states thus
formalize a distinction between two notions: the distributional behavior of the
oracle, which may be characterized in advance, and the particular latent
condition active in the current run, which must be inferred from the transcript.

In the context of LLMs, the distributional behavior corresponds to what a system
engineer can learn by probing the model before deployment---measuring response
frequencies, error rates, and tendencies across query types. A concrete example
of hidden states arises in mixture-of-experts (MoE)
architectures~\cite{Shazeer2017,Fedus2022}, in which a query is routed by a
gating mechanism to one or more expert sub-networks depending on the input, the
current query, and the prior transcript $t_{i-1}$. The routing decision is
hidden from the caller, who observes only the response. The system engineer may estimate the gating distribution over expert combinations
in advance, but cannot observe which experts fired for any specific
query---the latent condition active in a specific run must be inferred from the
response.

The hidden state in the SOTM framework represents this unobserved routing
condition. Each possible hidden state $b$ is associated with a response
distribution $\Dist_{q_i}^{b}$; the active hidden state determines which of
these distributions the oracle draws from, but the caller observes only the
response, not which state was active. Different expert combinations have
different weight matrices and therefore different response distributions,
including different probabilities of producing useful or correct responses.

On turn $i$, let $b_i$ denote the hidden state active when query $q_i$ is
issued in context $(x,t_{i-1})$. The oracle response $R_i$ is drawn from
$\Dist_{q_i}^{b_i}$, the response distribution under hidden state $b_i$. The
full random transcript is $\mathcal{T}_N=(Q_1,R_1,\ldots,Q_N,R_N)$, where $N=N(x)$
is the random number of turns. The response distributions $\Dist_{q_i}^{b_i}$
are not directly observable, but they can in principle be estimated through
repeated sampling under controlled conditions. We assume the system engineer has
characterized them, whether from calibration, benchmarking, or the task's
construction. What cannot be characterized in advance is which state $b_i$ is
active for a specific instance at hand---this is unknown per instance and must be
inferred from the response. The ceilings in
Sections~\ref{sec:identification} and~\ref{sec:coverage} are consequences of
this per-instance uncertainty; distributional uncertainty---not having
characterized the $\Dist_{q_i}^{b_i}$---is the subject of reliability
certification~\cite{Wang2026cert} and is not the focus of this paper.

In the MoE setting, let $B_i$ denote the random variable for the hidden state
active on turn $i$, taking values in the set of possible expert combinations.
The \textbf{gating distribution} is the conditional distribution of $B_i$ given
the context $(x,t_{i-1},q_i)$, with probabilities
$\pi_{b_i}=\Prb[B_i=b_i\mid x,t_{i-1},q_i]$ for each possible hidden state
$b_i$. Standard routing intends determinism---a fixed expert combination
selected by $(x,t_{i-1})$---but load balancing, numerical precision, and
intentional randomization make stochastic routing the more honest model, with
deterministic routing as the degenerate case where $\pi_{b_i}$ is a point mass
on a single expert combination. When the gating distribution is known, the
marginal response distribution
$\bar\Dist_{q_i} = \sum_{b_i} \pi_{b_i}\Dist_{q_i}^{b_i}$ is fully
characterized.

\section{Identification under Cached Responses}
\label{sec:identification}

This section studies binary identification under a cached-response oracle, in
the case where the score function is not incorporated into the SOTM. The hidden
state $B\in\{0,1\}$ has distribution $\pi_b=\Prb[B=b]$ for $b\in\{0,1\}$,
with $\pi_0+\pi_1=1$: output $1$ is correct under $B=1$, and output $0$ is
correct under $B=0$. The SOTM's output is an estimate $\hat{B}\in\{0,1\}$, and
success is evaluated externally by whether $\hat{B}=B$.

Without loss of generality, we assume the SOTM never reissues a query: if its
computation would repeat a query $q$ already asked, it reuses the stored
response rather than calling the oracle again. Any SOTM that reissues a query
can be replaced by an equivalent SOTM that stores and reuses the first response
internally. This guarantees that each distinct query contributes exactly one
response to the transcript. Write $P^b_{\mathcal{M}}(\cdot)$ for the
probability mass function of the transcript $\mathcal{T}_N$ under $B=b$, so
that $P^b_{\mathcal{M}}(t)$ is the probability that the computation produces
realization $t$.

\begin{theorem}[Identification by an adaptive cached-response SOTM]
\label{thm:single-sample-ceiling}
For a binary identification task under a cached-response oracle with distribution
$(\pi_0,\pi_1)$ on $B$, the supremum success probability over all adaptive
SOTMs is
\[
  \bar{\theta}
  =
  \sup_{\mathcal{M}}
  \left(\frac12 + \TV(\pi_1P^1_{\mathcal{M}},\pi_0P^0_{\mathcal{M}})\right),
\]
where the sum in $\TV$ ranges over all transcript realizations $t$.
For each fixed $\mathcal{M}$, the optimal decision rule attains the
optimal success probability: output $1$ when
$\pi_1P^1_{\mathcal{M}}(t)>\pi_0P^0_{\mathcal{M}}(t)$ and output $0$ otherwise,
with ties broken arbitrarily. This rule depends on $x$ through the transcript
distributions $P^b_{\mathcal{M}}$.
\end{theorem}

\begin{proof}
By the caching construction, each distinct realized query $q_i$ appearing in
$t$ contributes one response, drawn from $\Dist_{q_i}^{b_i}$ on its first and
only call, while repeated queries reuse the stored response. Thus, on the
probability-one halting event, $\mathcal{M}$ produces one finite transcript
realization $t$ determined by $M$'s internal randomness and the hidden states
$b_i$ under $B=b$.

This single-realization property is specific to cached responses: under fresh
responses, repeated calls to the same query return independent draws, and the
analysis would apply to the growing product distribution of independent draws.

Given a realization $t$, any decision rule $\hat{B}(x,t)\in\{0,1\}$ induces a
partition of transcripts for this fixed $x$: $A=\{t:\hat{B}(x,t)=1\}$ and
$A^c=\{t:\hat{B}(x,t)=0\}$; we write $\hat{B}(t)$ when $x$ is clear from
context. It succeeds when $\hat{B}(\mathcal{T})=B$ and errs when
$\hat{B}(\mathcal{T})\neq B$.
With distribution $(\pi_0,\pi_1)$ on $B$, its error probability is
\[
  P_e(\hat{B})
  =
  \pi_0\Prb_{\mathcal{T}\sim P^0_{\mathcal{M}}}[\hat{B}(\mathcal{T})=1]
  +
  \pi_1\Prb_{\mathcal{T}\sim P^1_{\mathcal{M}}}[\hat{B}(\mathcal{T})=0]
  =
  \pi_0\sum_{t\in A}P^0_{\mathcal{M}}(t)
  +
  \pi_1\sum_{t\in A^c}P^1_{\mathcal{M}}(t).
\]
For each transcript $t$, assigning $t$ to $A$ contributes
$\pi_0P^0_{\mathcal{M}}(t)$ to the error, while assigning it to $A^c$
contributes $\pi_1P^1_{\mathcal{M}}(t)$. Minimizing over all decision rules
gives the per-realization minimum
\[
  P_e^\ast(\mathcal{M})
  =
  \sum_t
  \min\{\pi_0P^0_{\mathcal{M}}(t),\,\pi_1P^1_{\mathcal{M}}(t)\},
\]
attained by the optimal decision rule, which places $t\in A$ exactly when
$\pi_1P^1_{\mathcal{M}}(t)>\pi_0P^0_{\mathcal{M}}(t)$ for this fixed $x$.

The optimal success probability for a fixed SOTM $\mathcal{M}$ is
\[
  1-P_e^\ast(\mathcal{M})
  =
  1-\sum_t
  \min\{\pi_0P^0_{\mathcal{M}}(t),\,\pi_1P^1_{\mathcal{M}}(t)\}.
\]
Using $\min\{a,b\}=\frac12(a+b-|a-b|)$ and
$\sum_tP^0_{\mathcal{M}}(t)=\sum_tP^1_{\mathcal{M}}(t)=1$, we obtain
\[
  P_e^\ast(\mathcal{M})
  =
  \frac12\left(
    1-\sum_t
    \left|\pi_1P^1_{\mathcal{M}}(t)-\pi_0P^0_{\mathcal{M}}(t)\right|
  \right).
\]
Since
\[
  \TV(\pi_1P^1_{\mathcal{M}},\pi_0P^0_{\mathcal{M}})
  =
  \frac12\sum_t
  \left|\pi_1P^1_{\mathcal{M}}(t)-\pi_0P^0_{\mathcal{M}}(t)\right|,
\]
the optimal success probability can be written as
\[
  1-P_e^\ast(\mathcal{M})
  =
  \frac12+\TV(\pi_1P^1_{\mathcal{M}},\pi_0P^0_{\mathcal{M}}).
\]
Taking the supremum over adaptive SOTMs $\mathcal{M}$ gives the stated
$\bar{\theta}$. Since every SOTM's output is determined by its transcript
realization, no SOTM can exceed the optimal success probability.
\end{proof}

Under equal probabilities $\pi_0=\pi_1=\frac12$, the success probability takes
the form $\frac12+\frac12\TV(P^1_{\mathcal{M}},P^0_{\mathcal{M}})$.

\begin{corollary}[Uniform distribution on $B$]
\label{cor:equal-priors}
If $\pi_0=\pi_1=\frac12$, the supremum success probability is
\[
  \bar{\theta}
  = \sup_{\mathcal{M}}\;\frac{1+\TV(P^1_{\mathcal{M}},P^0_{\mathcal{M}})}{2}.
\]
\end{corollary}

\begin{proof}
Set $\pi_0=\pi_1=\frac12$ in Theorem~\ref{thm:single-sample-ceiling}. For each
SOTM $\mathcal{M}$,
\[
  \TV\left(\frac12P^1_{\mathcal{M}},\frac12P^0_{\mathcal{M}}\right)
  = \frac12\TV(P^1_{\mathcal{M}},P^0_{\mathcal{M}}),
\]
by homogeneity of total variation. Substituting this into
Theorem~\ref{thm:single-sample-ceiling} gives
\[
  \bar{\theta}
  = \sup_{\mathcal{M}}\left(\frac12 +
    \frac12\TV(P^1_{\mathcal{M}},P^0_{\mathcal{M}})\right)
  = \sup_{\mathcal{M}}\;\frac{1+\TV(P^1_{\mathcal{M}},P^0_{\mathcal{M}})}{2}.
\]
\end{proof}

Many identification tasks have a natural decisive query: a query whose answer,
if reliable, would determine which of two candidate outputs will be judged
correct. The single-informative-query setting isolates this case: the
distinguished query $q^\ast$ is the only query whose response distribution
depends on the hidden state $B$. Under cached responses, this is the only useful
state-dependent response for identifying the true binary state; repeating
$q^\ast$ returns the same cached response, and all other queries are
uninformative. Thus the adaptive cached-response problem reduces to a one-sample
binary test.

\begin{corollary}[Single informative query under cached responses]
\label{cor:single-informative}
Suppose there is a distinguished query $q^\ast$ such that the response to
$q^\ast$ has distribution $\Dist^0$ under $B=0$ and $\Dist^1$ under $B=1$,
while every query $q\neq q^\ast$ has a response distribution independent of $B$. Then one call to
$q^\ast$ is
sufficient to attain the optimal cached-response success probability, and
additional calls to $q^\ast$ or to any $q\neq q^\ast$ cannot improve it. The
supremum success probability is
\[
  \bar{\theta}
  =
  1-\sum_r \min\{\pi_0\Dist^0(r),\,\pi_1\Dist^1(r)\}.
\]
Under equal probabilities $\pi_0=\pi_1=\frac12$, let $d=\TV(\Dist^1,\Dist^0)$.
Then
\[
  \bar{\theta}
  =
  \frac{1+d}{2}.
\]
\end{corollary}

\begin{proof}
By assumption, all queries $q\neq q^\ast$ have the same response distribution
under $B=0$ and $B=1$, so their responses do not help distinguish the hidden
state. By caching, repeated calls to $q^\ast$ return the same response already
seen on the first call. Hence any transcript generated by an adaptive SOTM
contains no state-dependent information beyond the single cached response
$r\sim\Dist^b$ from $q^\ast$.

Therefore the induced binary testing problem reduces to testing $\Dist^0$
against $\Dist^1$ from one observation $r$. Applying
Theorem~\ref{thm:single-sample-ceiling} gives
\[
  \bar{\theta}
  =
  1-\sum_r \min\{\pi_0\Dist^0(r),\,\pi_1\Dist^1(r)\}.
\]
This value is attained by the SOTM that queries $q^\ast$ once and applies the
optimal decision rule to the cached response $r$. Under equal probabilities
$\pi_0=\pi_1=\frac12$, Corollary~\ref{cor:equal-priors} gives
$\bar{\theta}=(1+\TV(\Dist^1,\Dist^0))/2=(1+d)/2$.
\end{proof}

\begin{remark}[Multiple informative queries under cached responses]
The single-informative-query corollary isolates the cleanest cached-response
case. More generally, suppose the binary hidden state is $B\in\{0,1\}$ and the
SOTM can issue several informative queries $q_1,\ldots,q_k$. For query $q_i$,
let $\Dist_i^b$ denote the response distribution under $B=b$. If the SOTM issues
these queries and receives responses $R_1,\ldots,R_k$, then, conditional on $B=b$, the responses $R_1,\ldots,R_k$ are independent with
response vector having distribution
\[
  \left(\bigotimes_{i=1}^k\Dist_i^b\right)(r_1,\ldots,r_k)
  =
  \prod_{i=1}^k \Dist_i^b(r_i).
\]
Theorem~\ref{thm:single-sample-ceiling} then applies with transcript laws
induced by these product distributions. Under equal probabilities
$\pi_0=\pi_1=\frac12$ and a fixed set of queries, the success probability is
\[
  \frac{1+\TV\!\left(
    \bigotimes_{i=1}^k\Dist_i^1,
    \bigotimes_{i=1}^k\Dist_i^0
  \right)}{2}.
\]
Thus cached responses do not rule out evidence accumulation across distinct
informative queries; they rule out gaining new evidence by repeating the same
query. Adaptivity matters because the SOTM may choose later queries based on
earlier responses.
\end{remark}

Theorem~\ref{thm:single-sample-ceiling} extends naturally from two hidden states
to any finite hidden-state space, although the resulting expression no longer
reduces to a single total-variation distance.

\begin{theorem}[Multi-state identification by an adaptive cached-response SOTM]
\label{thm:multi-state-identification}
Let $\mathcal{B}$ be a finite hidden-state space with distribution
$(\pi_b)_{b\in\mathcal{B}}$, where $\pi_b=\Prb[B=b]$ and $\sum_b\pi_b=1$.
Exactly one state $B\in\mathcal{B}$ is realized per computation. The SOTM
succeeds if its final decision $\hat B$ equals the realized state. The
supremum success probability over all adaptive SOTMs is
\[
  \bar{\theta}_{\mathcal{B}}
  =
  \sup_{\mathcal{M}}
  \sum_t \max_{b\in\mathcal{B}} \pi_b P^b_{\mathcal{M}}(t).
\]
For each fixed $\mathcal{M}$, the optimal final decision rule chooses any state
\[
  \hat B(t)\in
  \arg\max_{b\in\mathcal{B}} \pi_b P^b_{\mathcal{M}}(t).
\]
\end{theorem}

\begin{proof}
Fix an adaptive SOTM $\mathcal{M}$. A final decision rule partitions transcript
realizations into sets $A_b=\{t:\hat B(t)=b\}$, one for each
$b\in\mathcal{B}$. Its success probability is
\[
  \sum_{b\in\mathcal{B}} \pi_b \sum_{t\in A_b} P^b_{\mathcal{M}}(t).
\]
For each transcript realization $t$, assigning $t$ to state $b$ contributes
$\pi_b P^b_{\mathcal{M}}(t)$ to the success probability. Hence the optimal
assignment chooses the state with the largest contribution
$\pi_b P^b_{\mathcal{M}}(t)$, and the optimal success probability for this
fixed $\mathcal{M}$ is
\[
  P_{\rm succ}^\ast(\mathcal{M})
  =
  \sum_t \max_{b\in\mathcal{B}} \pi_b P^b_{\mathcal{M}}(t).
\]
Taking the supremum over adaptive SOTMs gives the claim.
\end{proof}

\section{Output Quality under Cached Responses}
\label{sec:coverage}

Identification issues concern ambiguity about the hidden state. Output quality
issues concern a different question: after receiving its cached responses, how
good an output can the SOTM produce? The answer is evaluated by a score function
$S:X\times Y\to[0,1]$, so the SOTM need not determine the hidden state exactly.
Several outputs may be useful, and usefulness is measured by their score.

Under cached responses, the SOTM may still issue many queries. The constraint is
that each distinct query can contribute at most one new response: repeating the
same query returns the cached response already received. Thus the relevant
object is the finite halting transcript produced by the SOTM, together with the
set of outputs the SOTM can compute from that transcript.

Fix an instance $x$. Let
$\mathcal{T}$ denote the finite halting transcript produced by a
cached-response SOTM $\mathcal{M}$ on input $x$, and let $P_{\mathcal{M},x}$
be its distribution. For a transcript realization $t$, write
$\mathcal{Y}_{\mathcal{M}}(x,t)$ for the set of outputs in $Y$ computable from
the input $x$ and transcript $t$.

For a fixed SOTM and realized transcript, the best score available from that
transcript is
\[
  \sup_{y\in\mathcal{Y}_{\mathcal{M}}(x,t)} S(x,y).
\]
The output quality ceiling is the expected value of this best available score,
optimized over SOTMs.

\begin{remark}[Why ceiling]
We use the term ceiling because the expression is tight: it is the best expected
quality achievable under the stated information constraint, not merely a loose
upper bound. No SOTM can exceed it, and the supremum is approached by SOTMs that
adaptively query to maximize the score of their computable outputs.
\end{remark}

\begin{theorem}[Output quality under cached responses]
\label{thm:coverage-ceiling}
Under cached responses, the output quality ceiling is
\[
  \bar\theta
  =
  \sup_{\mathcal{M}}
  \E_{x\sim\Dist_X}
  \E_{\mathcal{T}\sim P_{\mathcal{M},x}}
  \left[
    \sup_{y\in\mathcal{Y}_{\mathcal{M}}(x,\mathcal{T})} S(x,y)
  \right].
\]
\end{theorem}

\begin{proof}
Fix an SOTM $\mathcal{M}$ and an input $x$. The cached-response oracle,
together with the machine's internal randomness, induces a distribution
$P_{\mathcal{M},x}$ on finite halting transcripts. Let $t$ be one such
transcript realization. After transcript $t$ is produced, the SOTM's final
output must belong to $\mathcal{Y}_{\mathcal{M}}(x,t)$. Therefore, on the
realized pair $(x,t)$, no output computed from that transcript can have score
larger than $\sup_{y\in\mathcal{Y}_{\mathcal{M}}(x,t)} S(x,y)$, which is the
best score achievable from that transcript realization. Thus, for this fixed
$\mathcal{M}$, the expected output quality is
\[
  \E_{x\sim\Dist_X}
  \E_{\mathcal{T}\sim P_{\mathcal{M},x}}
  \left[
    \sup_{y\in\mathcal{Y}_{\mathcal{M}}(x,\mathcal{T})} S(x,y)
  \right].
\]
Taking the supremum over cached-response SOTMs $\mathcal{M}$ gives the stated
output quality ceiling.
\end{proof}

In the single-informative-query case, the only state-dependent information in
the transcript is the single cached response $r$ to the distinguished query
$q^\ast$. Let $\Dist_x^b$ be the response distribution of this cached response
under hidden state $B=b$, and let $\bar\Dist_x=\sum_b\pi_b\Dist_x^b$ be the
corresponding averaged response distribution.

\begin{definition}[Inter-turn computation and single-response expected quality]
\label{def:single-response-value}
An \textbf{inter-turn computation} (Definition~\ref{def:sotm}) applied to the
single-informative-query setting is a map
$\phi:X\times(\Sigma^\ast)^\ast\to\Sigma^\ast$ that takes the input $x$ and the
realized transcript $t_i=(q_1,r_1,\ldots,q_i,r_i)$ after turn $i$, and returns
either the next query $q_{i+1}\in\Sigma^\ast$ or a halting output $y\in
Y\subseteq\Sigma^\ast$. In the single-informative-query
case, where $q^\ast$ is the only state-dependent query, the
\textbf{single-response expected quality} of $\phi$ on instance $x$ is
\[
  V(\phi,x)
  =
  \E_{r\sim\bar\Dist_x}\bigl[S(x,\phi(x,(q^\ast,r)))\bigr].
\]
\end{definition}

\begin{corollary}[Single-informative-query output quality]
\label{cor:single-informative-coverage}
Under the single-informative-query assumption, with the oracle as the sole source
of information beyond $x$, the output quality ceiling reduces to
\[
  \bar\theta
  =
  \sup_\phi\;
  \E_{x\sim\Dist_X}\bigl[V(\phi,x)\bigr],
\]
where $\phi$ ranges over inter-turn computations.
\end{corollary}

\begin{proof}
Under the single-informative-query assumption, the SOTM may issue multiple
queries, but all queries other than $q^\ast$ supply no information useful for
improving the score of the final output beyond what the SOTM can compute from
$x$ alone. Repeating $q^\ast$ returns the same cached response $r$, adding no
new information. Thus the useful part of the transcript for output quality is
the single cached response $r\sim\bar\Dist_x$ to $q^\ast$, and the SOTM's
remaining computation can be represented by an inter-turn computation
$\phi(x,(q^\ast,r))$. Substituting this reduced transcript into
Theorem~\ref{thm:coverage-ceiling} gives the displayed formula by
Definition~\ref{def:single-response-value}.
\end{proof}

\begin{corollary}[Single-informative output quality with acceptance threshold]
\label{cor:coverage-threshold}
Fix an acceptance threshold $\eta\in[0,1]$; call an output $y$
\textbf{acceptable} if $S(x,y)\ge\eta$. In the single-informative-query setting,
for each inter-turn computation $\phi$, define its output quality frontier by
\[
  F_\phi(\eta)
  =
  \E_{x\sim\Dist_X}
  \Prb_{r\sim\bar\Dist_x}\bigl[S(x,\phi(x,(q^\ast,r)))\ge\eta\bigr].
\]
Then $F_\phi$ is non-increasing in $\eta$, and for every $\eta>0$,
\[
  F_\phi(\eta)
  \le
  \frac{1}{\eta}
  \E_{x\sim\Dist_X}\bigl[V(\phi,x)\bigr].
\]
Moreover, the single-informative output quality ceiling is
\[
  \bar\theta
  =
  \sup_\phi
  \int_0^1 F_\phi(\eta)\,d\eta .
\]
\end{corollary}

\begin{proof}
Fix an inter-turn computation $\phi$. By definition,
\[
  F_\phi(\eta)
  =
  \E_x
  \Prb_r\bigl[S(x,\phi(x,(q^\ast,r)))\ge\eta\bigr]
  =
  \Prb_{x,r}\bigl[S(x,\phi(x,(q^\ast,r)))\ge\eta\bigr],
\]
where $x\sim\Dist_X$ and $r\sim\bar\Dist_x$. Since
$S(x,\phi(x,(q^\ast,r)))$ is a nonnegative random variable, Markov's inequality gives,
for every $\eta>0$,
\[
  F_\phi(\eta)
  \le
  \frac{1}{\eta}
  \E_{x,r}\bigl[S(x,\phi(x,(q^\ast,r)))\bigr]
  =
  \frac{1}{\eta}
  \E_{x\sim\Dist_X}\bigl[V(\phi,x)\bigr].
\]
The function $F_\phi$ is non-increasing because increasing $\eta$ can only
shrink the event $\{S(x,\phi(x,(q^\ast,r)))\ge\eta\}$.

Finally, since $S(x,\phi(x,(q^\ast,r)))\in[0,1]$, the layer-cake identity gives
\[
  \E_{x,r}\bigl[S(x,\phi(x,(q^\ast,r)))\bigr]
  =
  \int_0^1
  \Prb_{x,r}\bigl[S(x,\phi(x,(q^\ast,r)))\ge\eta\bigr]\,d\eta
  =
  \int_0^1 F_\phi(\eta)\,d\eta .
\]
Taking the supremum over inter-turn computations $\phi$ and using
Corollary~\ref{cor:single-informative-coverage} gives the stated expression for
$\bar\theta$.
\end{proof}

\begin{remark}[Output quality generalizes competence]
For a $0/1$ score function, $S(x,\phi(x,(q^\ast,r)))\in\{0,1\}$ and
\[
  \E_r[S(x,\phi(x,(q^\ast,r)))]=\Prb[S(x,\phi(x,(q^\ast,r)))=1],
\]
so the single-informative output quality ceiling simplifies to the maximum
probability of a correct output from one cached state-dependent response. With
$\phi(x,(q^\ast,r))=r$ (the SOTM outputs the response directly), this equals
$q(x)$, the per-input competence of Definition~\ref{def:competence}. The output
quality ceiling thus extends competence from a binary correctness probability to
an expected score.
\end{remark}

The identification and output quality results isolate two distinct issues that
cached responses can pose: ambiguity about which output is correct, and lack of
enough information to produce a high-scoring output.

In the single-informative equal-probability case, the identification ceiling $(1+d)/2$
quantifies the first issue. The transcript must be used to determine which
output is correct, and total variation distance measures the irreducible
confusion between the two hidden states. Even a high-quality, articulate response
is useless if it is equally consistent with both states.

The output quality ceiling quantifies the second issue. The transcript must be used
to compute an output in $Y$, and the expected score of that computation
measures how much quality the cached responses can supply. Even a perfectly
unambiguous transcript is useless if no computation using it produces a
high-scoring output.

The identification and output quality ceilings cap what cached responses can
achieve, and both are raised by fresh responses---identification through the
accumulation of evidence at the rate given by the Chernoff information, output
quality through repeated draws until an acceptable output appears.
Section~\ref{sec:fresh} explains how fresh responses raise these two ceilings and
quantifies the number of queries required to do so.

\begin{remark}[Overlap affects identification ceilings]
For identification, the operative quantity is the overlap between the
state-conditional transcript distributions, not whether the response space is
finite or infinite. Greater overlap means smaller total variation distance, and
therefore a lower identification ceiling. If the two transcript distributions
are mutually singular, then the hidden state can be recovered perfectly from the transcript and the
identification ceiling equals $1$. If they overlap, then some transcript
realizations remain compatible with both hidden states, producing a strict
ceiling below perfect success. Thus unbounded support is neither necessary nor
sufficient for an identification ceiling: two distributions on an infinite set
can be mutually singular, while two distributions on a two-point set can overlap.
\end{remark}

\section{Fresh Responses and Query Counts}
\label{sec:fresh}

Under cached responses, repeated calls to the same query return no new
information. Under fresh responses, each call produces an independent response,
so the SOTM can accumulate evidence about the hidden state or keep searching for
an acceptable output. This section first quantifies how repeated fresh responses
to an informative query raise the identification ceiling, and then asks how many
queries are needed to reach a target quality level.

\subsection{Fresh-Response Identification with a Single Informative Query}

We first consider the single-informative-query specialization of the binary
hidden-state identification problem in Section~\ref{sec:identification},
paralleling Corollary~\ref{cor:single-informative}. The general adaptive
problem may involve several informative queries; here we isolate the case in
which one distinguished query $q^\ast$ is the only query whose response
distribution depends on the hidden state. Under hidden state $B=b$, each fresh
call to $q^\ast$ returns an independent response with distribution $\Dist^b$.
For clarity, we take the two hidden states to have equal probabilities in this
subsection. If the SOTM makes $k$ fresh calls to $q^\ast$, then its informative
transcript is
\[
  \mathcal{T}_k=(q^\ast,R_1,q^\ast,R_2,\ldots,q^\ast,R_k),
\]
with response vector $R^k=(R_1,\ldots,R_k)$ distributed as
\[
  R^k\mid B=b \sim (\Dist^b)^{\otimes k},
  \qquad
  (\Dist^b)^{\otimes k}(r^k)
  =
  \prod_{i=1}^k \Dist^b(r_i).
\]
Responses to any other queries are independent of $B$ and therefore do not
improve the optimal identification probability. Thus fresh responses allow the
SOTM to accumulate evidence about the true binary state.

Let $C(\Dist^0,\Dist^1)$ denote the Chernoff information between the two
response distributions:
\[
  C(\Dist^0,\Dist^1)
  =
  -\log \inf_{0\le s\le1}\sum_r \Dist^0(r)^s \Dist^1(r)^{1-s}.
\]
It is the sharp exponential decay rate of the optimal error probability for
testing $\Dist^0$ against $\Dist^1$ from i.i.d.\ samples~\cite{CoverThomas2006}.

\begin{theorem}[Fresh-response identification and Chernoff information]
\label{thm:fresh-response-identification}
In the single-informative-query fresh-response setting of this subsection, with
equal probabilities $\pi_0=\pi_1=\frac12$ on $B$, let $P_e^\star(k)$ be the
minimum error probability after $k$ fresh calls to $q^\ast$. Then
\[
  P_e^\star(k)
  =
  \frac12\left(
    1-\TV\!\left((\Dist^1)^{\otimes k},(\Dist^0)^{\otimes k}\right)
  \right).
\]
If $\Dist^0\neq\Dist^1$, then
\[
  -\lim_{k\to\infty}\frac1k\log P_e^\star(k)
  =
  C(\Dist^0,\Dist^1).
\]
Since $C(\Dist^0,\Dist^1)>0$ whenever $\Dist^0\neq\Dist^1$, we have
$P_e^\star(k)\to0$ and the optimal success probability tends to one.
\end{theorem}

\begin{proof}
Theorem~\ref{thm:single-sample-ceiling} gives the optimal success probability
for any pair of transcript distributions. Under cached responses, repeating a query returns the same response and therefore
does not produce a product distribution.
Under fresh responses, repeated calls to $q^\ast$ produce independent responses,
so after $k$ calls the state-dependent transcript distributions are
$(\Dist^0)^{\otimes k}$ and $(\Dist^1)^{\otimes k}$. Applying
Theorem~\ref{thm:single-sample-ceiling} to these fresh-response transcript
distributions with equal probabilities, as in Corollary~\ref{cor:equal-priors}, gives
the optimal success probability
\[
  \bar{\theta}_k
  =
  1-P_e^\star(k)
  =
  \frac{1+\TV\!\left((\Dist^1)^{\otimes k},(\Dist^0)^{\otimes k}\right)}{2}.
\]
Equivalently,
\[
  P_e^\star(k)
  =
  \frac12\left(
    1-\TV\!\left((\Dist^1)^{\otimes k},(\Dist^0)^{\otimes k}\right)
  \right).
\]
For i.i.d. samples from two distributions, the optimal error probability has
exponential decay rate equal to the Chernoff information~\cite{CoverThomas2006}.
Since $\Dist^0\neq\Dist^1$ implies $C(\Dist^0,\Dist^1)>0$, the error tends to
zero and $\bar{\theta}_k\to1$.
\end{proof}

\subsection{Fresh-Response Query Counts with an Incorporated Score Function}

Sections~\ref{sec:identification} and~\ref{sec:coverage} characterize the
identification and output quality ceilings under cached responses. The
corresponding fresh-response question is how many queries an SOTM needs to reach
a target quality level. The primary object in the bounds below is the query count
$N(x)$, the number of fresh oracle calls made on input $x$.
The quality condition over $x\sim\Dist_X$ is
\[
  \E_{x\sim\Dist_X}[S(x,\mathcal{M}(x))]\ge\theta,
\]
where the expectation is over $x\sim\Dist_X$, the oracle's response
distributions, and $M$'s internal randomness.

When the score function is incorporated into the SOTM, fresh responses can
convert the output quality ceiling into a stopping problem. Let $\phi$ be an
inter-turn computation as in Definition~\ref{def:single-response-value}.
Given a fresh response $r$, the SOTM computes $y=\phi(x,(q^\ast,r))$ and,
because the score function is incorporated, can evaluate whether $S(x,y)\ge\theta$.
For target quality~$\theta$,
define the one-response acceptable-output probability
\[
  p_\phi(x,\theta)
  =
  \Prb_{r}\bigl[S(x,\phi(x,(q^\ast,r)))\ge\theta\bigr].
\]
If $p_\phi(x,\theta)>0$, then repeated independent responses make it possible
for an SOTM with access to the score function to wait for an acceptable output.
Thus fresh responses do not make every output high scoring, but they can make an
acceptable output appear with probability approaching one whenever the
acceptable-response probability is positive. The result below makes this query
count explicit. For a $0/1$ score function and $\theta=1$,
threshold-acceptable means correct.

\begin{theorem}[Threshold stopping with an incorporated score function]
\label{thm:verifier}
Fix $\theta\in(0,1]$ and an inter-turn computation $\phi$ as in
Definition~\ref{def:single-response-value}, and suppose $p_\phi(x,\theta)>0$
for $\Dist_X$-almost every $x$. The SOTM that repeatedly queries the
fresh-response oracle on $x$ until it computes an output with score at
least~$\theta$ has $N(x)$ geometrically distributed with mean
$1/p_\phi(x,\theta)$. The expected query count, allowing the value $+\infty$,
is
\[
  \E_{x\sim\Dist_X}\!\left[N(x)\right]
  =
  \E_{x\sim\Dist_X}\!\left[\frac{1}{p_\phi(x,\theta)}\right].
\]
\end{theorem}

\begin{proof}
Construct the SOTM as follows. It repeatedly queries the fresh-response oracle.
After receiving response $r_i$, it computes the candidate output
$y_i=\phi(x,(q^\ast,r_i))$ and evaluates $S(x,y_i)$. If $S(x,y_i)\ge\theta$, it halts and
outputs $y_i$; otherwise, it issues another query. Under fresh responses, each
query-response turn produces a threshold-acceptable output independently with
probability $p_\phi(x,\theta)$. The first threshold-acceptable output therefore
has a geometrically distributed query count:
\[
  \Prb[N(x)=k]=(1-p_\phi(x,\theta))^{k-1}p_\phi(x,\theta),
  \qquad
  \E[N(x)]=\frac{1}{p_\phi(x,\theta)}.
\]
For the fixed input $x$, the expectation is over the fresh oracle responses.
Averaging this identity over $x\sim\Dist_X$ gives the displayed formula.
\end{proof}

When the quality target $\theta<1$, it may be necessary to stop without waiting
until a threshold-acceptable output appears: for example, a query budget,
computation time limit, or token budget may prevent indefinite fresh queries. We
therefore allow a \emph{query-budget cutoff}: after a prescribed number of fresh
oracle calls, the SOTM halts even if no threshold-acceptable output has appeared.

\begin{theorem}[Threshold stopping with a query-budget cutoff]
\label{thm:query-budget-cutoff}
Fix an input $x$ and let $p=p_\phi(x,\theta)>0$. If the threshold-stopping SOTM
is cut off after $m$ fresh oracle calls, then its success probability is
\[
  1-(1-p)^m.
\]
To guarantee success probability at least a target $\gamma\in(0,1)$ on this
input, the least sufficient cutoff is
\[
  m_\gamma(x)
  =
  \left\lceil
    \frac{\log(1-\gamma)}{\log(1-p_\phi(x,\theta))}
  \right\rceil .
\]
Under a cutoff at $m$ calls, the expected query count is
\[
  \E[N_m(x)]
  =
  \frac{1-(1-p)^m}{p},
\]
where $N_m(x)$ is the number of queries made before success or cutoff.
\end{theorem}

\begin{proof}
For the fixed input $x$, each fresh oracle call produces a threshold-acceptable
output independently with probability $p=p_\phi(x,\theta)$. With a cutoff after
$m$ calls, success occurs unless all $m$ calls fail, so the success probability
is $1-(1-p)^m$. The condition $1-(1-p)^m\ge\gamma$ is equivalent to
\[
  m\ge \frac{\log(1-\gamma)}{\log(1-p)},
\]
because $\log(1-p)<0$, which gives the stated integer cutoff. The query count
under the cutoff is $N_m(x)=\min\{N(x),m\}$, so $N_m(x)$ is a positive
integer-valued random variable bounded by $m$. For any such random variable
$Z$,
\[
  Z=\sum_{j=1}^{m}\mathbf{1}\{Z\ge j\}.
\]
Taking expectations gives the tail-sum formula
$\E[Z]=\sum_{j=1}^{m}\Prb[Z\ge j]$. Applying this formula to $Z=N_m(x)$,
\[
  \E[N_m(x)]
  =
  \sum_{j=1}^{m}\Prb[N_m(x)\ge j]
  =
  \sum_{j=1}^{m}\Prb[N(x)\ge j]
  =
  \sum_{j=1}^{m}(1-p)^{j-1}
  =
  \frac{1-(1-p)^m}{p}.
\]
\end{proof}

\subsection{Fresh-Response Query Counts without the Score Function}

Without the score function, no SOTM can verify during its run that a computed
output is threshold-acceptable. The full output quality problem without the score
function may involve arbitrary computations from the response transcript, so we
isolate the simplest output quality setting in which fresh responses can still
help. We call this the \emph{binary candidate-output model}, referring to the two
possible responses $y^\ast$ and $\tilde{y}$. In this model, each
fresh oracle response is itself a candidate output: for a fixed input $x$, it
equals a threshold-acceptable output $y^\ast$ with probability $q(x)$ and a
fixed unacceptable alternative $\tilde y$ with probability $1-q(x)$,
independently across calls. The SOTM observes the response sequence but does not
have access to the score function, so it cannot verify during the computation
which candidate is acceptable. It can, however, compare received candidate
outputs for equality and count their frequencies. The theorem below treats the
favorable case $q(x)>1/2$, where the acceptable output is the more likely of the two responses.

Operationally, the SOTM does not check candidate outputs during the run. It
makes a fixed number $k$ of fresh oracle calls, records the candidate outputs it
receives, and returns the empirical majority candidate. This majority rule is
justified by the assumption $q(x)>1/2$: because the acceptable output
$y^\ast$ is more likely than the unacceptable alternative $\tilde y$,
repeated fresh responses make $y^\ast$ the majority candidate with high
probability. The score function or ground truth is used only externally, after
the run, to define and analyze whether the returned output is successful; it is
not available to the SOTM while it is computing.

\begin{theorem}[Fixed-sample majority without the score function]
\label{thm:no-verifier}
In the binary candidate-output model, fix $x$ with $q(x)>1/2$. Let
$k_\theta(x)$ be the least fixed number of fresh queries needed for
empirical-majority output to reach success probability at least $\theta$. Then,
to leading order as $\theta\to1$,
\[
  k_\theta(x)
  =
  \frac{\ln(1/(1-\theta))}{D(0.5\|q(x))}
  \cdot (1+o(1)).
\]
\end{theorem}

\begin{proof}
With $k$ queries, the number of threshold-acceptable responses is
$B_k\sim\mathrm{Binomial}(k,q(x))$. Empirical-majority output is
threshold-acceptable exactly when $B_k>k/2$, up to a tie convention that affects
only lower-order terms. For
$q(x)>1/2$, the error probability is $\Prb[B_k\le k/2]$. The Chernoff large-deviation estimate for binomial
tails~\cite{CoverThomas2006} gives, for
$Z_k\sim\operatorname{Binomial}(k,q)$ and $a<q$,
\[
  \Prb[Z_k\le ka]
  =
  e^{-k\,D(a\|q)(1+o(1))}
  \qquad\text{as }k\to\infty,
\]
where $D(a\|q)$ is the Bernoulli KL divergence. Here
$a=1/2$, and the assumption $q(x)>1/2$ ensures $a<q(x)$, so the estimate gives
\[
  \Prb[B_k\le k/2]
  =
  e^{-k\,D(0.5\|q(x))\,(1+o(1))}.
\]
Setting $e^{-k\,D(0.5\|q(x))(1+o(1))}=1-\theta$ and solving for $k$
gives the displayed query count.
\end{proof}

\begin{remark}[When the acceptable output is less probable]
The condition $q(x)>1/2$ is essential for majority amplification in this binary
candidate-output model. If $q(x)<1/2$, the unacceptable response is more
probable, so empirical majority converges to the unacceptable candidate. If
$q(x)=1/2$, fresh responses carry no directional information about which
candidate is acceptable. Thus, without access to the score function or other
side information identifying the acceptable candidate, repeated fresh responses
cannot drive the success probability to one when $q(x)\le1/2$.
\end{remark}

A useful special case is when the success probabilities are constant across
inputs.

\begin{corollary}[Uniform success probabilities]
\label{cor:uniform}
If $p_\phi(x,\theta)=p$ for $x\sim\Dist_X$, then with the score function
incorporated the threshold-stopping query count is
\[
  \E_{x\sim\Dist_X}[N(x)]
  =
  \frac{1}{p}.
\]
In the binary candidate-output model without the score function with $q>1/2$, the
least fixed query count for empirical-majority output is
\[
  k_\theta
  =
  \frac{\ln(1/(1-\theta))}{D(0.5\|q)}
  \cdot(1+o(1)).
\]
\end{corollary}

\begin{proof}
Both formulas follow by setting $p_\phi(x,\theta)=p$ and $q(x)=q$ in
Theorems~\ref{thm:verifier} and~\ref{thm:no-verifier} respectively.
\end{proof}

The results in this section identify two ways fresh responses can help. When the
score function is incorporated into the SOTM, repeated oracle calls allow the
SOTM to keep querying until it computes a threshold-acceptable output and then
halts, possibly with a query-budget cutoff. When the score function is not
available during computation, the SOTM cannot recognize success as it occurs; in
the binary candidate-output model, repeated oracle calls instead help by making
the threshold-acceptable candidate appear as the empirical majority with high
probability.

\section{Conclusion}
\label{sec:conclusion}

This paper developed information-theoretic ceilings and query-count bounds for
computation with stochastic oracles. Under cached responses, repeated calls to
the same query do not provide new information. The resulting transcript therefore
imposes two distinct ceilings. The identification ceiling is governed by the
total variation distance between transcript distributions: when the distributions
induced by two hidden states overlap, some transcript realizations remain
compatible with both states, and no SOTM can distinguish them perfectly from the
transcript alone. The output quality ceiling is governed by the best expected
score obtainable from the transcript: even when the output need not
identify a hidden state, the responses may not contain enough information to
produce a high-quality output.

Fresh responses raise both ceilings by returning an independent response on each
call. For identification, repeated fresh calls reduce the identification error
exponentially at the Chernoff rate
(Theorem~\ref{thm:fresh-response-identification}). For output quality with an
incorporated score function, fresh responses turn the problem into threshold
stopping: the expected query count equals the reciprocal of the probability of
obtaining a threshold-acceptable output (Theorem~\ref{thm:verifier}), with a
query-budget cutoff available when indefinite querying is not feasible
(Theorem~\ref{thm:query-budget-cutoff}). Without access to the score function,
fresh responses can still help in special settings, such as the binary
candidate-output model, where majority voting identifies the more likely
acceptable candidate with high probability (Theorem~\ref{thm:no-verifier}). The
identification ceiling also extends to finite hidden-state spaces beyond the
binary case (Theorem~\ref{thm:multi-state-identification}), and the
single-informative output quality ceiling admits an explicit frontier
characterization via an acceptance threshold
(Corollary~\ref{cor:coverage-threshold}).

Several extensions remain open. First, the query-count bounds with an
incorporated score function assume that $S(x,y)$ is computed without error. In
practical LLM systems, score functions may be learned evaluators, automated
tests, or verifiers subject to false positives and false negatives. Extending
the query-count bounds to noisy score functions would require replacing
$p_\phi(x,\theta)$ by an effective acceptable-output probability that accounts
for evaluation error.

Second, the paper gives closed-form identification results for the
single-informative-query setting, while the general cached-response problem may
involve several informative queries. In that case, the relevant object is the
joint transcript distribution induced by the chosen adaptive query policy.
Characterizing how token cost trades off against the adaptive selection of
informative queries is a natural next step.

Third, Theorem~\ref{thm:no-verifier} treats the binary candidate-output model
without the score function. In that setting, the SOTM makes a fixed number of
fresh oracle calls and returns the candidate output that appears most often; the
analysis shows that this empirical majority is threshold-acceptable with high
probability when the acceptable candidate is more likely than the unacceptable
alternative. General tasks without an incorporated score function may have many
incorrect responses, semantic equivalence classes, or graded partial scores. The
right replacement for the Bernoulli KL divergence $D(0.5\|q)$ should involve an
information radius or a multi-hypothesis large-deviation exponent.

Finally, the analysis treats the oracle as a fixed family of response
distributions. Prompt changes, model updates, retrieval augmentation, and
correlated errors across inputs all alter the oracle's response distributions.
These changes do not
invalidate the results, but they change the oracle to which the results apply.
Understanding such drift is important for applying stochastic-oracle bounds to
deployed AI systems.

\bibliographystyle{plainnat}
\bibliography{reference}

\end{document}